\RequirePackage{fix-cm}

\documentclass[smallextended]{svjour3}

\usepackage{amssymb}
\usepackage{graphicx}
\usepackage{epstopdf}
\usepackage{subfigure}
\usepackage{amsmath}
\usepackage{bm}
\usepackage{braket}

\journalname{Quantum Information Processing}
\begin{document}

\title{Global multipartite entanglement dynamics in Grover's search algorithm}

\author{Minghua Pan \and Daowen Qiu \and Shenggen Zheng}
\institute{Minghua Pan \at
               School of Electronics and Information Technology, Sun Yat-sen University, Guangzhou 510006, China\\
               School of Information and Electronic Engineering, Wuzhou University, Wuzhou 543002, China\\
   \and Daowen Qiu \and Shenggen Zheng \at
   Institute of Computer Science Theory, School of Data and Computer Science, Sun Yat-sen University, Guangzhou 510006, China \\
  \email{issqdw@mail.sysu.edu.cn} (Corresponding author) \\
 }

\date{Received: date / Accepted: date}

\maketitle

\begin{abstract}
Entanglement is considered to be one of the primary reasons for why quantum algorithms are more efficient than their classical counterparts for certain computational tasks.
The global multipartite entanglement of the multiqubit states in Grover's search algorithm can be quantified using the geometric measure of entanglement (GME). Rossi {\em et al.} (Phys. Rev. A \textbf{87}, 022331 (2013)) found that the entanglement dynamics is scale invariant for large $n$. Namely, the GME does not depend on the number $n$ of qubits; rather, it only depends on the ratio of iteration $k$ to the total iteration. In this paper, we discuss the optimization of the GME for large $n$. We prove that ``the GME is scale invariant'' does not always hold. We show that there is generally a turning point that can be computed in terms of the number of marked states and their Hamming weights during the curve of the GME. The GME is scale invariant prior to the turning point. However, the GME is not scale invariant after the turning point since it also depends on $n$ and the marked states. 

\end{abstract}

\section{Introduction}
\label{introduction}
Entanglement is a unique feature of quantum theory, and it is considered to be one of the key resources in quantum computation and information \cite{book1,JMP43}. For quantum computation, it has been demonstrated that quantum algorithms are more efficient than classical algorithms for certain computational tasks, such as Shor's factoring algorithm \cite{Shor} and Grover's search algorithm \cite{Grover97}. Entanglement is believed to be the main reason for the efficiency in these quantum algorithms \cite{PRSLA459,PRL91}. It has been shown that entanglement is necessary to achieve an exponential speedup in Shor's factoring algorithm \cite{PRSLA459}. In the Deutsch-Jozsa \cite{PRSLA439}, Grover, and Simon algorithms \cite{Simon94}, it has been shown that multipartite entanglement is involved for most instances \cite{PRA83}. Because of its quadratic speedup over the best classical algorithm, the role of entanglement in Grover's search algorithm has attracted considerable interest \cite{PRL85,PRA65,PRA75,IJQI6,JMP43,PRA69,PLA345,PLA373,JPMT43,PRA87,13054454,Nat14,QIP152,QIP1511}.
Entanglement plays a critical role in Grover's dynamic search algorithm even though the initial state and the final state are separable \cite{PRA69,13054454}. Without entanglement, the quantum search algorithm would require an exponential overhead in terms of resources \cite{PRL85}. The correlations in Grover's algorithm were discussed in \cite{JPMT43} for one marked state, including concurrence as the entanglement measure. It was shown \cite{PLA373} that Grover's algorithm can be described as an iteration change of the bipartite entanglement in terms of concurrence. The multipartite
entanglement feature of the quantum states was investigated using the separable degree
\cite{Nat14}. It has been shown \cite{JMP439,PRA69,PLA345,PLA373,JPMT43,PRA87,13054454,QIP152} that the curve of the entanglement during Grover's search algorithm first increases and then decreases, which means that a turning point exists. From an algebraic geometry perspective, Holwech {\em et al.} \cite{QIP1511} investigated the entanglement nature of quantum states generated by Grover's search algorithm and explained the turning point of the curve when the marked states are $\ket{00...0}$ and Greenberger-Horne-Zeilinger (GHZ) states \cite{07120921}.

Geometric entanglement was originally defined as the Euclidean distance of a given multipartite state to the nearest fully separable state \cite{ANYA755}. The geometric measure of entanglement (GME) is a suitable entanglement measure when multipartite systems are taken into account \cite{PRA68}. The GME is a global quantifier of entanglement that includes bipartite and multipartite contributions \cite{PRA77}. Moreover, the GME is an interesting quantifier because it has connections with other measures \cite{PRA73} and can be efficiently estimated by quantitative entanglement witnesses that are amenable to experimental verification \cite{PRL98}.
According to the expression of the GME, Shantanav {\em et al.} \cite{13054454} numerically demonstrated how the GME changes with $n$ qubits (the database size was $N=2^n$) and $M$ marked states in Grover's algorithm.
Rossi {\em et al.} \cite{PRA87} presented explicit expressions of the GME in the asymptotic limit $2^n\gg1$ when $M=1$ and $M=2$ (GHZ state). These authors showed that the GME was independent of the number of qubits, $i.e.$, scale invariance for large $n$ in these two cases.

These previous works motivate us to consider how to compute the turning point and whether there is scale invariance of entanglement dynamics for general $M$ marked states for large $N$.
To this end, we discuss the optimization of the GME for $M$ marked states in a database with $N=2^n$ items when $N\gg M$. We find that the turning point of the GME dynamics curve can be computed in terms of the number $n$ of qubits and the marked states. Prior to the turning point, the GME is scale invariant. However, it is generally not scale invariant after the turning point, which may depend on $n$ and the marked states.

The remainder of this paper is organized as follows. In Section \ref{Sec2}, we briefly overview the related background and the previous works about entanglement dynamics in terms of the geometric measure of entanglement (GME). In Section \ref{Sec3}, we study the amount of entanglement dynamically evolving and provide the analytical asymptotic expression of the entanglement for symmetric states \cite{PRA80} by optimizing the GME.
In Section \ref{Sec4}, we present the entanglement dynamics for representative symmetric states to more clearly show our conclusions. Finally, we discuss our results and present a conclusion in Section \ref{Sec5}.

\section{Preliminaries}\label{Sec2}
\subsection{Geometric measure of entanglement in Grover's algorithm}
The original Grover's algorithm searches for a target ({\em i.e.}, the marked state) in an unordered database with $N$ entries using $O(\sqrt N)$ queries, which requires $\Omega(N)$ queries even for the best  classical randomized algorithm. In this paper, we consider a system with $N$ items and $M$ marked states, where $N=2^{n}$. To clarify the notation, we will simply overview Grover's search algorithm; for more details, see \cite{book1,Grover97}. Grover's search algorithm starts with an $n$-qubit uniform superposition state
\begin{eqnarray}\label{eqpsi0}
\ket{\psi_0}=\frac{1}{\sqrt{N}}\sum_{x\in\{0,1\}^n}\ket{x}.
\end{eqnarray}

Define $\ket{S_0}\equiv\frac{1}{\sqrt{N-M}}\sum_{x_n}\ket{x_n}$
as the superposition of all the states $\ket{x_n}$ that are not marked states, and
$\ket{S_1}\equiv\frac{1}{\sqrt{M}}\sum_{x_s}\ket{x_s}$
represents the superposition of all the states $\ket{x_s}$ that are marked states
({\em i.e.}, the solutions of search problem). It is easy to show that $\ket{\psi_0}$ can be rewritten as
\begin{eqnarray}\label{eqpsi}
\ket{\psi_0}&&=\sqrt{\frac{N-M}{N}}\ket{S_0}+\sqrt{\frac{M}{N}}\ket{S_1}
=\cos\theta\ket{S_0}+\sin\theta\ket{S_1},
\end{eqnarray}
where $\theta=\arcsin\sqrt{M/N}$. Then, the Grover operation $G$ (also referred to as the Grover iteration) will be applied to $\ket{\psi_0}$. After $k$ iterations of the Grover operation $G$, we obtain the state
\begin{eqnarray}
\ket{\psi_{k,M}} \equiv G^k\ket{\psi_0}= \cos\theta_k\ket{S_0}+\sin\theta_k\ket{S_1},
\end{eqnarray}
where $\theta_k=(2k+1)\theta$.

Operation $G$ is repeated until state $\ket{\psi_{k,M}}$ overlaps with $\ket{S_1}$ to the greatest extent possible. It is proven that the optimal total iteration is $k_{opt}=\mbox{CI}[{\frac{\pi}{2\theta}-1}/2]$, where $\mbox{CI}[x]$ denotes the closest integer to $x$. Ideally, the final state $\ket{\psi_{k_{opt}}}$ prior to measurement is $\ket{S_1}$, which is the superposition of all the marked states.
When $M\ll N$, we have $k_{opt}=\lfloor\frac{\pi}{4}\sqrt{N/M}\rfloor$.
Clearly, the optimal total iteration $k_{opt}$ is $O(\sqrt N)$, and $\theta_k=\frac{\pi}{2}k/k_{opt}$. It is clear that $\theta_k$ only depends on $k/k_{opt}$.

The geometric measure of entanglement (GME) \cite{ANYA755} of a state $\ket{\psi}$ is expressed as its distance from its nearest separable state $\ket{\phi}$. Namely, the overlap between
$\ket{\psi}$ and $\ket{\phi}$ is maximized. The entanglement of state $\ket{\psi}$ is
\begin{eqnarray}\label{defEn}
E(\ket{\psi})=1-\max_\phi|\braket{\psi|\phi}|^2,
\end{eqnarray}
where the maximum is over all separable states ($i.e.$, $\ket{\phi}=\bigotimes_{s=1}^n\ket{\phi_s}$, where the states $\ket{\phi_s}$ are single-qubit pure states).
The geometric measure remains unknown for most of the multipartite states \cite{PRA81} because the definition involves an optimization procedure over the class of separable states. However, this task can be drastically simplified in the case where states are symmetric since the nearest product state to any symmetric multipartite quantum state is necessarily symmetric \cite{PRA80}.

Fortunately, a large number of quantum states in experiments are symmetric under particle exchange, and this property allows us to significantly reduce the computational complexity.
For an $n$-particle system, the state $\ket{\tilde{m}}$ is defined as the following unnormalized symmetric state \cite{PRA67}:
\begin{eqnarray}\label{mexci}
\ket{\tilde{m}}\equiv\sum_i P_i(\ket{1}^{\otimes m}\ket{0}^{\otimes n-m}),
\end{eqnarray}
where $P_i$ is the set of all $n \choose m$ distinct permutations of $m$ $1$s and $n-m$ $0$s.

Therefore, we can use the GME to quantify the global multipartite entanglement of the state generated by Grover's search algorithm if the marked state $\ket{S_1}$ is symmetric according to the following fact.

\newtheorem{fact}{\textbf{Fact}}\label{Fa1}
\begin{fact}
For Grover's algorithm, let $\ket{S_1}$ be the superposition of all the marked states. If $\ket{S_1}$ is symmetric, then the state after $k$ iterations $\ket{\psi_{k,M}}$ is symmetric for $k\in \{0,1,...,k_{opt}\}$, where $M$ is the number of marked states and $k_{opt}$ is the optimal total iteration.
\end{fact}
\begin{proof}
Let $P$ be any particle exchange ($i.e.$, qubit permutation) operation. Clearly, the initial uniform superposition state is symmetric, that is,
$P\ket{\psi_0}=\ket{\psi_0}$.
Suppose that the marked state $\ket{S_1}$ is symmetric; we have
$P\ket{S_1}=\ket{S_1}$. According to Eq. (\ref{eqpsi}), we have $\ket{S_0}=\frac{1}{\cos\theta}\ket{\psi_0}-\tan\theta\ket{S_1}$. For any intermediate state $\ket{\psi_{k,M}}$, we have
\begin{eqnarray}
P\ket{\psi_{k,M}}&&=P(\cos\theta_k\ket{S_0}+\sin\theta_k\ket{S_1})\nonumber\\
&&=P(\frac{\cos\theta_k}{\cos\theta}\ket{\psi_0}+(\sin\theta_k-\cos\theta_k\tan\theta)\ket{S_1})\nonumber\\
&&=\frac{\cos\theta_k}{\cos\theta}P\ket{\psi_0}+(\sin\theta_k-\cos\theta_k\tan\theta)P\ket{S_1})\nonumber\\
&&=\frac{\cos\theta_k}{\cos\theta}\ket{\psi_0}+(\sin\theta_k-\cos\theta_k\tan\theta)\ket{S_1})\nonumber\\
&&=\ket{\psi_{k,M}}, \nonumber
\end{eqnarray}
where $k=0,1,...,k_{opt}$.
Therefore, $\ket{\psi_{k,M}}$ is symmetric.
\end{proof}

\subsection{Previous works about global GME in Grover's algorithm}
The optimization of the GME can be performed on the restricted set of symmetric separable states $\ket\eta...\ket\eta$, where $\ket\eta=\cos\frac{\alpha}{2}\ket{0}+e^{i\beta}\sin\frac{\alpha}{2}\ket{1}$.
The maximization involves only two parameters: $\alpha\in[0,\pi]$ and $\beta\in[0,2\pi]$.
Since $\theta_k\in[0,\frac{\pi}{2}]$, the coefficients of $\ket{\psi_{k,M}}$ are all positive. Therefore, the phase factor can be fixed to $\beta=0$.

For simplicity, we denote $|x|$ as the Hamming weight of a state $\ket{x}$, which is the number of $1$s in $x\in\{0,1\}^n$. Let $n_1,n_2,...,n_M $ be the Hamming weights of $M$ marked states.
The symmetric $n$-separable state can be expressed as
$\ket{\phi}=\ket{\eta}^{\bigotimes n}=\sum_{x\in \{0,1\}^n} \cos^{n-|x|}\frac{\alpha}{2}\sin^{|x|}\frac{\alpha}{2}\ket{x}$.
According to \cite{13054454}, the overlap of state $\ket{\psi_{k,M}}$ and the separate state $\ket{\phi}$ is
\begin{eqnarray}\label{epEn}
\braket{\psi_{k,M}|\phi}=&&\left(\cos\theta_k\bra{S_0}+\sin\theta_k\bra{S_1}\right)
\left(\sum_{x\in \{0,1\}^n}\cos^{n-|x|}\frac{\alpha}{2}\sin^{|x|}\frac{\alpha}{2}\ket{x}\right)\nonumber\\
=&&\frac{\cos\theta_k}{\sqrt{N-M}}\left[\left(\cos{\frac{\alpha}{2}}+\sin{\frac{\alpha}{2}}\right)^n
 -\sum_{i=1}^M\cos^{n-n_i}\frac{\alpha}{2}\sin^{n_i}\frac{\alpha}{2}\right]\nonumber\\
 &&+\frac{\sin\theta_k}{\sqrt{M}}\sum_{i=1}^M \cos^{n-n_i}\frac{\alpha}{2}\sin^{n_i}\frac{\alpha}{2}.
\end{eqnarray}
Therefore, the maximum of the overlap between state $\ket{\psi_{k,M}}$ and the $n$-separate state is
\begin{eqnarray}\label{rf1}
\max_\phi|\braket{\psi_{k,M}|\phi}|^2=
&&\max_\alpha\left|\frac{\cos\theta_k}{\sqrt{N-M}}\left[\left(\cos{\frac{\alpha}{2}}+\sin{\frac{\alpha}{2}}\right)^n
-\sum_{i=1}^M\cos^{n-n_i}\frac{\alpha}{2}\sin^{n_i}\frac{\alpha}{2}\right]\right.\nonumber\\
&&\left. +\frac{\sin\theta_k}{\sqrt{M}}\sum_{i=1}^M \cos^{n-n_i}\frac{\alpha}{2}\sin^{n_i}\frac{\alpha}{2}\right|^2.
\end{eqnarray}

It is difficult to perform optimization without knowledge of $n$ and the marked states ($n_i,M$).
Shantanav {\em et al.} \cite{13054454} numerically calculated the GME for various $n$ and $M$. These authors also discussed the cases where the marked states were GHZ states, Dicke states \cite{Dicke} and $W$ states \cite{PRA62}.
From another perspective, Rossi {\em et al.} \cite{PRA87} discussed the explicit expressions of the GME in large $N$ when the numbers of marked states are $M=1$ and $M=2$ (the marked states are $\ket{00...0}$ and $\ket{11...1}$), and they found ``scale invariance of entanglement dynamics''. Namely, the GME depends only on $\theta_k\simeq \frac{\pi}{2}k/k_{opt}$ and not on separate $k$ and $n$. In fact, the GMEs of these cases have two parts. The first parts are both asymptotic to $\sin^2\theta_k$. The second parts are asymptotic to
$\cos^2\theta_k$ (for $M=1$) and $\frac{1+\cos^2\theta_k}{2}$ (for $M=2$).
The turning points between the two parts are $k_{opt}/2$ and $0.61 k_{opt}$ for these two cases. From an algebraic geometry perspective, Holwech {\em et al.} \cite{QIP1511}  geometrically explained the fact that the turning point corresponds to $k_{opt}/2$ (for $M=1$) and $2k_{opt}/3$ (for GHZ states).

\section{Entanglement dynamics in the asymptotic limit}\label{Sec3}
We now consider that the algorithm searches in a large database with $M$ marked states such that $N\gg M$. Since $|\sum_{i=1}^M( \cos^{n-n_i}\frac{\alpha}{2}\sin^{n_i}\frac{\alpha}{2})|\leq M$, we have
\begin{eqnarray}
\max_\alpha\left|\frac{\cos\theta_k}{\sqrt{N-M}}\sum_{i=1}^M \cos^{n-n_i}\frac{\alpha}{2}\sin^{n_i}\frac{\alpha}{2}\right|^2\simeq 0. \nonumber
\end{eqnarray}

Therefore, we can obtain
\begin{eqnarray}\label{rf2}
\max_\phi|\braket{\psi_{k,M}|\phi}|^2\simeq
&&\max_\alpha\left|\frac{\cos\theta_k}{\sqrt{N-M}}\left(\cos{\frac{\alpha}{2}}+\sin{\frac{\alpha}{2}}\right)^n\right.\nonumber\\
&&\left.+\frac{\sin\theta_k}{\sqrt{M}}\sum_{i=1}^M \cos^{n-n_i}\frac{\alpha}{2}\sin^{n_i}\frac{\alpha}{2}\right|^2.
\end{eqnarray}

For simplicity, define $\Lambda\equiv\braket{\psi_{k,M}|\phi}$.
The GME can be expressed as $E(\ket{\psi_{k,M}})=1-\max|\Lambda|^2$.
It is clear that $(\cos{\frac{\alpha}{2}} + \sin{\frac{\alpha}{2}})^{n}=
\sqrt{2^n}\sin^n({\frac{\pi}{4}+\frac{\alpha}{2}})$.
Denote $A(\alpha)=\sin^n({\frac{\pi}{4}+\frac{\alpha}{2}})$ and $B(\alpha)=\frac{1}{\sqrt M}\sum_{i=1}^{M}\cos^{n-n_{i}}\frac{\alpha}{2}\sin^{n_{i}}\frac{\alpha}{2}$. We have
\begin{eqnarray}\label{Lambda}
\Lambda=\braket{\psi_{k,M}|\phi}\simeq \cos\theta_{k}A(\alpha)+\sin\theta_{k}B(\alpha)
\end{eqnarray}
and Eq. (\ref{rf2}) can be rewritten as
\begin{eqnarray}\label{AB}
\max|\Lambda|^2\simeq \max_\alpha\left| \cos\theta_{k}A(\alpha)
+\sin\theta_{k}B(\alpha)\right|^2.
\end{eqnarray}

For symmetric states, the GMEs are the same when the Hamming weights are $|x|$ and $|n-x|$. Hence, we can consider only the cases where $n_i$ is from $0$ to $\frac{n}{2}$.
The optimization is still difficult to perform because $B(\alpha)$ cannot be analytically expressed without fixing $n$, $n_i$ and $M$. Prior to obtaining an overall optimization of Eq. (\ref{AB}), we will first consider the optimizations of $\cos\theta_{k}A(\alpha)$ and $\sin\theta_{k}B(\alpha)$. Since $\cos\theta_{k}$ and $\sin\theta_{k}$ are independent of $\alpha$, we can  focus only on the optimizations of $A(\alpha)$ and $B(\alpha)$.

\subsection{The optimization of $A(\alpha)$}
When $\alpha=\pi/2$, it is clear that $A(\alpha)$ reaches its maximum  $A_{max}=\max_\alpha(\sin^n(\frac{\pi}{4}+\frac{\alpha}{2}))=1$, and $B=\sqrt{\frac{M}{N}}$.
Since $N\gg M$, we have $B=\sqrt{\frac{M}{N}}\simeq 0$. When $\alpha$ is
away from $\pi/2$, $A(\alpha)$ will decay exponentially to zero.
Therefore, we can consider that $A(\alpha)$ has a non-zero value only when $\alpha$ is near $\pi/2$.

\subsection{The optimization of $B(\alpha)$}
Denote $D(\alpha)=\cos^{n-n_i}\frac{\alpha}{2}\sin^{n_i}\frac{\alpha}{2}$. We have $B(\alpha)=\frac{1}{\sqrt M}\max_\alpha(\sum_{i=1}^M D(\alpha))$.
The maximum of $D(\alpha)$ is
\begin{eqnarray}\label{Dmax}
D_{\max}=\left({\frac{n-n_i}{n}}\right)^{\frac{n-n_i}{2}}\left(\frac{n_i}{n}\right)^{\frac{n_i}{2}},
\end{eqnarray}
which is obtained when $\alpha=2\arccos\sqrt{\frac{n-n_i}{n}}$.
When $n$ is fixed, $D_{\max}$ decreases from $1$ to $1/\sqrt N$ as $n_i$ varies from $0$ to $n/2$. Since $M\ll N$, $n_i$ must be small. Therefore, the optimal $D_{max}$ is close to $1$, and $\alpha$ is near zero.

When taking $M$ marked states into consideration, the optimization will be more complicated.
For simplicity, we consider the same $n_i$, $i.e.$, the Hamming weights of all the marked states are the same. In this case, we have $M=(^n_{n_i})$ and $B=\sqrt M \cos^{n-n_i}\frac{\alpha}{2}\sin^{n_i}\frac{\alpha}{2}$. Therefore, the optimization of $B(\alpha)$ is
\begin{eqnarray}
B_{max}&&=\sqrt M \max_\alpha \left(\cos^{n-n_i}\frac{\alpha}{2}\sin^{n_i}\frac{\alpha}{2}\right)=\sqrt M D_{\max}.
\end{eqnarray}
When $n_i$ is small, the optimization of $B(\alpha)$ would make $A(\alpha)$ tend to $\frac{1}{\sqrt N}$, which is asymptotic to $0$ when $N\gg M$.

\vskip\baselineskip
According to the above discussion, we can summarize that the optimizations of $A(\alpha)$ and $B(\alpha)$ are restricted to each other when $N\gg M$.
Fig. \ref{fig:n100} shows that the values of $A(\alpha)$ and $B(\alpha)$ change with $\alpha$ for different $n_i$ when $n=100$.
Suppose that $g(\alpha)=A(\alpha)+B(\alpha)$. When $n_i$ is small ($n_i/n<1/10$ in Fig. \ref{fig:n100}), we can observe that $g(\alpha)$ is asymptotically a sectional function that is composed of $A(\alpha)$ and $B(\alpha)$. Namely,
\begin{eqnarray}\label{fs}
g(\alpha)\simeq
\left\{
\begin{aligned}
&A(\alpha),&&0\leq \alpha/2< \varepsilon;\\
&B(\alpha),&&\frac{\pi}{4}-\varepsilon'\leq \alpha/2 <\frac{\pi}{4}+\varepsilon';\\
&0,&&otherwise,\\
\end{aligned}
\right.
\end{eqnarray}
where $\varepsilon,\varepsilon'\geq 0$ and $\varepsilon+\varepsilon'\leq\pi/4$. Clearly, the optimizations of $A(\alpha)$ and $B(\alpha)$ are restricted to each other. In other words, the optimization of $g(\alpha)$ can only be obtained when $A(\alpha)$ or $B(\alpha)$ is optimized.
Note that the optimization of $g(\alpha)$ cannot be achieved by optimizing $A(\alpha)$ or $B(\alpha)$ separately without the assumption that $N\gg M$.
\begin{figure}
\centerline{\includegraphics[width=0.8\textwidth]{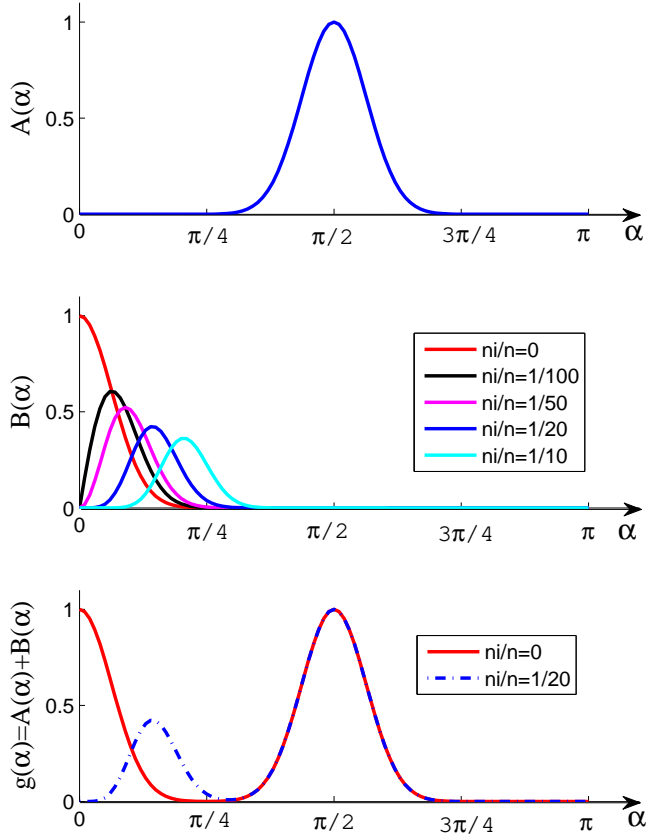}}
\caption{\label{fig:n100}(Color online) The values of $A(\alpha)=\sin^n({\frac{\pi}{4}+\frac{\alpha}{2}})$, $B(\alpha)=\sqrt M \cos^{n-n_i}\frac{\alpha}{2}\sin^{n_i}\frac{\alpha}{2}$ and $g(\alpha)=A(\alpha)+B(\alpha)$ change with $\alpha$ for $n=100$.}
\end{figure}

\subsection{The optimization of GME}
For simplicity, we would like to present the following lemmas.
\newtheorem{Lemma}{\textbf{Lemma}}\label{L1}
\begin{Lemma}
The maximum overlap of $\ket{\psi_{k,M}}$ and $\ket\phi$ is
$\max_\phi|\braket{\psi_{k,M}|\phi}|^2=\max(\cos^2\theta_k,\sin^2\theta_k B_{\max}^2)$ when $N\gg M$.
\end{Lemma}
\begin{proof}
Let $f(\alpha)=A(\alpha)+\tan\theta_k B(\alpha)$.
We have $\Lambda=\cos\theta_k f(\alpha)$ and $\max_\alpha |\Lambda|= \cos\theta_k \max_\alpha |f(\alpha)$| as $\cos\theta_k\geq 0$.
According to Eq. (\ref{fs}), it is easy to obtain
$\max|f(\alpha)|={\max(\max |A(\alpha)|, \tan\theta_k\cdot\max |B(\alpha)|)}=\max(1,\tan\theta_k B_{\max})$ for $\theta_k\in [0,\pi/2]$ and $\tan\theta_k\geq 0$.
Therefore, we have $\max|\Lambda|=\max(\cos\theta_k,\sin\theta_k B_{\max})$.
Since $\max_\phi|\braket{\psi_{k,M}|\phi}|^2=\max|\Lambda|^2$, the lemma follows.
\end{proof}

\newtheorem{Lemma2}{\textbf{Lemma}}\label{L2}
\begin{Lemma}
A turning point $\theta_{k_T}$ exists in the GME curve during Grover's search algorithm when $N\gg M$, which can be computed by the number $n$ of qubits and the number of marked states and their Hamming weights in general cases.
\end{Lemma}

\begin{proof}
According to Lemma $1$, $\max|\Lambda|=\max(\cos\theta_k,\sin\theta_k B_{max})$. If $\tan \theta_k\leq 1/B_{max} $, then $\max|\Lambda|=\cos\theta_k$. Otherwise, $\max |\Lambda|=  \sin\theta_k B_{max}$.

Case $1$: If $\tan \theta_k\leq 1/B_{max} $, then $\max|\Lambda|=\cos\theta_k$. The GME is
\begin{equation}
E(\ket{\psi_{k,M}})\simeq \sin^2\theta_k,
\end{equation}
which is actually equal to the success probability of the search algorithm. Since $\sin^2\theta_k$ is increasing monotonously for $\theta_k\in[0,\pi/2]$, $E(\ket{\psi_{k,M}})$ is a monotonic increasing function.

Case $2$: If $\tan \theta_k>1/B_{max} $, we have $\max|\Lambda|=\sin\theta_k B_{max}$. The maximum of the overlap is
\begin{eqnarray}\label{partB}
\max_\phi|\braket{\psi_{k,M}|\phi}|^2
&&\simeq \sin^2\theta_k \max_\alpha |B|^2\nonumber\\
&&=\frac{\sin^2\theta_k}{M}\max_\alpha\left|\sum_{i=1}^M \cos^{n-n_i}\frac{\alpha}{2}\sin^{n_i}
\frac{\alpha}{2}\right|^2.
\end{eqnarray}
The maximum depends not only on $\theta_k$ but also on $M, n$ and $n_i$. The optimization $$\max_\alpha|B|^2=\frac{1}{M}\max_\alpha\left|\sum_{i=1}^M \cos^{n-n_i}\frac{\alpha}{2}\sin^{n_i}\frac{\alpha}{2}\right|^2$$
is fixed to be constant only in the case of $M,n,n_i$ being fixed. In such a case, Eq. (\ref{partB}) is proportional to $\sin^2\theta_k$, which is increasing monotonously as $\theta_k\in[0,\pi/2]$. Therefore, the GME is decreasing monotonously in this case.

According to the above discussion, we know that a turning point exists, which is the position of the maximum entanglement during the search algorithm. Hence, we have $\max|\Lambda|=\cos\theta_{k_T}=\sin\theta_{k_T} B_{\max}$ at the turning point.
Therefore, the turning point is
\begin{equation}\label{thk}
\theta_{k_T}=\arctan\left(\frac{1}{B_{max}}\right)=\arctan\left(\frac{\sqrt{M}}{\max_\alpha
\sum_{i=1}^M \cos^{n-n_i}\frac{\alpha}{2}\sin^{n_i}\frac{\alpha}{2}}\right),
\end{equation}
where $k_T=\frac{2}{\pi}\theta_{k_T}k_{opt}$ is the corresponding turning iteration.
According to the expression of $\theta_{k_T}$, we can obtain that the turning point typically depends on the number $n$ of qubits and the marked states (the number $M$ of marked states and their Hamming weights $n_i,i=1,2,...,M$).
\end{proof}

Furthermore, we can obtain the following theorem as the main result.
\newtheorem{theorem1}{\textbf{Theorem}}\label{Theorem1}
\begin{theorem}
For a given $N=2^n$ item database with $M$ marked states {\em (}{\em i.e.,} solutions{\em )}, let the global GME of  the state $\ket{\psi_{k,M}}$ after $k$ iterations in Grover's algorithm be $E(\ket{\psi_{k,M}})$. When $N\gg M$, the GME is
\begin{eqnarray}\label{Econ0}
E(\ket{\psi_{k,M}})\simeq
\left\{
\begin{aligned}
&\sin^2\theta_k,&\theta_k\leq\theta_{k_T};\nonumber\\
&1-\frac{\sin^2\theta_k}{M}\max_{\alpha\in[0,\pi]}\left|\sum_{i=1}^M
\cos^{n-n_i}\frac{\alpha}{2}\sin^{n_i}\frac{\alpha}{2}\right|^2,&\theta_k > \theta_{k_T},
\end{aligned}
\right.
\end{eqnarray}
where $n_i$ is the Hamming weight of the $i$th marked state, $\theta_k=(2k+1)\theta$,
$\theta_{k_T}$ is named the turning point, and $k_T$ is the corresponding turning iteration.
\end{theorem}
\begin{proof}
 According to Eq. (\ref{defEn}), Lemma $1$ and Lemma $2$, the theorem follows.
\end{proof}

According to Eq. (\ref{thk}) and Theorem $1$, we can clearly observe that the first part of the GME of state $\ket{\psi_{k,M}}$ has scale invariance because it is asymptotic to $\sin^2\theta_k$, which only depends on $k/k_{opt}$. However, during the entire process, it generally does not have scale invariance because it  depends not only on $k/k_{opt}$ but also on $n,n_i$ and $M$ as $\theta_k>\theta_{k_T}$.
Note that the Hamming weights $n_i$ of the marked states play an important role in the calculation of the GME.\\

Furthermore, we can obtain the following corollaries.
\newtheorem{corollary1}{\textbf{Corollary}}\label{Cor1}
\begin{corollary}
The entanglement of Grover's algorithm is scale invariant if and only if the optimization of $B(\alpha)=\frac{1}{\sqrt M}\sum_{i=1}^M \cos^{n-n_i}\frac{\alpha}{2}\sin^{n_i}\frac{\alpha}{2}$ is independent of separate $M, n, n_i$.
\end{corollary}
\begin{proof}
$(\Rightarrow)$ If the optimization of $B(\alpha)=\frac{1}{\sqrt M}\sum_{i=1}^M \cos^{n-n_i}\frac{\alpha}{2}\sin^{n_i}\frac{\alpha}{2}$ is independent of separate $M, n, n_i$, then $B_{max}=\lambda$, where $\lambda$ is a constant.
Thus, $E(\ket{\psi_{k,M}})\simeq 1-\lambda\sin^2\theta_k$ as $\theta_k > \theta_{k_T}$. Therefore, the entanglement only depends on $\theta_k$, $i.e.$, only depends on $k/k_{opt}$, which is scale invariant. \\
$(\Leftarrow)$ If the entanglement of Grover's algorithm is scale invariant, then both parts of the GME only depend on $\theta_k$. This means that the optimization of $B(\alpha)=\frac{1}{\sqrt M}\sum_{i=1}^M \cos^{n-n_i}\frac{\alpha}{2}\sin^{n_i}\frac{\alpha}{2}$ is constant, which is independent of separate $M, n, n_i$.
\end{proof}

Consider the case where the Hamming weights $n_i$ of the $M$ marked states are the same; in this case, we will have the following:
\newtheorem{corollary2}{\textbf{Corollary}}\label{Cor2}
\begin{corollary}
Suppose that the Hamming weights of all the marked states are identical; the GME of the state
$\ket{\psi_{k,M}}$ in Grover's algorithm is
\begin{eqnarray}\label{Econ}
E(\ket{\psi_{k,M}})\simeq
\left\{
\begin{aligned}
&\sin^2\theta_k,&\theta_k\leq\theta_{k_T};\nonumber\\
&1-\sin^2\theta_k M \left({1-\frac{n_i}{n}}\right)^{n-n_i}\left(\frac{n_i}{n}\right)^{n_i},&\theta_k > \theta_{k_T}.
\end{aligned}
\right.
\end{eqnarray}
\end{corollary}
\begin{proof}
When the Hamming weights of all the marked states are identical, {\em i.e.}, all $n_i$ are the same, we have $\sum_{i=1}^M \cos^{n-n_i}\frac{\alpha}{2}\sin^{n_i}\frac{\alpha}{2}
= M \cos^{n-n_i}\frac{\alpha}{2}\sin^{n_i}\frac{\alpha}{2}$.
According to Eq. (\ref{Dmax}), it is clear that $\max_\alpha|\cos^{n-n_i}\frac{\alpha}{2}\sin^{n_i}\frac{\alpha}{2}|^2=
({\frac{n-n_i}{n}})^{n-n_i}(\frac{n_i}{n})^{n_i}$ for $\alpha\in [0,\pi]$.
According to Theorem $1$, the corollary  follows.
\end{proof}

In this case, the turning point is
\begin{equation}\label{thk2}
\theta_{k_T}=\arctan\left(\frac{(1-{\frac{n_i}{n}})^{\frac{n_i-n}{2}}(\frac{n_i}{n})^{\frac{-n_i}{2}}}{\sqrt{M}}\right).
\end{equation}

\section{ Entanglement dynamics for representative symmetric states}\label{Sec4}
Let us assume that there are $M$ marked states in the system and that $\ket{S_{1}}=\frac{1}{\sqrt{M}}\sum_{x_{s}}\ket{x_{s}}$, where $|x_{s}\rangle$ is a marked state.  According to the above results,  if $|S_1\rangle$ is a symmetric state, then we are able to calculate the global multipartite dynamics during Grover's algorithm via the GME.
To further confirm that the optimization of the GME is reasonable and to show the above conclusions more clearly, we consider that $|S_1\rangle$ are some specific representative symmetric states. We will present the GMEs of the states $\ket{\psi_{k,M}}$, both analytic expressions in large $N$ and some numerical results.

\subsection{Product state}
Suppose that the state $\ket{S_{1}}$ is a product state ({\em i.e.}, fully $n$-separable state). Because local unitary operations do not change the entanglement in terms of GME \cite{PRA68}, product states can be changed into each other by applying local unitary operations that would not change the entanglement.
Without loss of generality, we can consider that $\ket{S_{1}}=\ket{00...0}$. The result can be generalized to other product marked states (irrespective of whether they are symmetric).
In this case, the marked state is $\ket{00...0}$, and the number of marked states is $M=1$.
Since the initial state is uniform, which is also a product state, the entanglement is zero both at the beginning and at the end of the search algorithm.
In this case, $B(\alpha)=\cos^n\frac{\alpha}{2}$, and the overlap of the state $\ket{\psi_{k,M=1}}$ and the
separable state $\ket{\phi}$ is
\begin{eqnarray}\label{M1lap}
\braket{\psi_{k,M=1}|\phi}=
\frac{\cos\theta_k}{\sqrt{N-1}}\left[\left(\cos{\frac{\alpha}{2}}+\sin{\frac{\alpha}{2}}\right)^n
 -\cos^n\frac{\alpha}{2}\right]
 +\sin\theta_k\cos^n\frac{\alpha}{2}.
\end{eqnarray}

Because $\max_\alpha(\cos^n\frac{\alpha}{2})=\max_\alpha(\cos\frac{\alpha}{2})=1$ when $\alpha=0$, the optimization of $B(\alpha)$ is independent of $n$. According to Eq. (\ref{thk}), we have $\theta_{k_T}=\pi/4$ and $k_T=0.5k_{opt}$.
Therefore, when $n\gg1$, according to Theorem $1$, the GME has the following simple form:
\begin{eqnarray}\label{M1r}
E(\ket{\psi_{k,M=1}})\simeq
\left\{
\begin{aligned}
&\sin^2\theta_k,&\theta_k\leq\pi/4;\\
&\cos^2\theta_k,&\theta_k >\pi/4.
\end{aligned}
\right.
\end{eqnarray}

\begin{figure}
\includegraphics[width=0.8\textwidth]{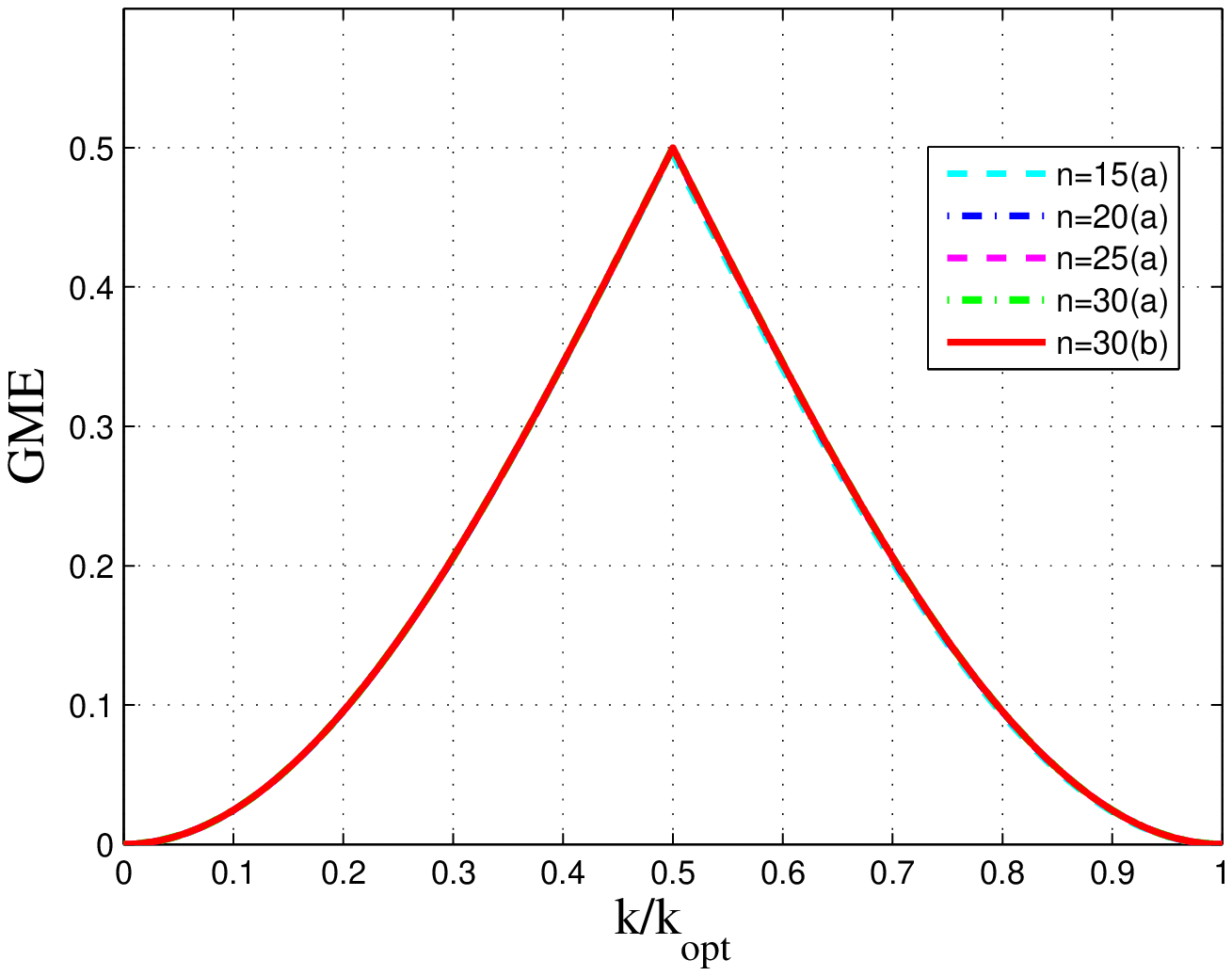}
\caption{\label{fig:M1}(Color online) Entanglement dynamics vs $k/k_{opt}$ when the marked states are $\ket{00...0}$ for $n=15\sim30$. Here, (a) means that the curves of GME  are plotted by Eq. (\ref{defEn}) and Eq. (\ref{M1lap}), and (b) is plotted by Eq. (\ref{M1r}).
}
\end{figure}

According to Eq. (\ref{M1r}), it is clear that the GME is scale invariant and only depends on
$\theta_k=\frac{\pi}{2}k/k_{opt}$.
For clarity, we present some numerical results in Fig. \ref{fig:M1}.
We can obtain that the curves only change with $k/k_{opt}$ but do not depend on $n$, showing scale invariance. Moreover, the curves of the GME plotted by Eq. (\ref{defEn}), Eq. (\ref{M1lap}) and Eq.( \ref{M1r}) are identical when $n=30$, which also shows that our method to optimize the GME is reasonable.

\newtheorem{Remark}{\textbf{Remark}}\label{R1}
\begin{Remark}
Since the optimization of $B(\alpha)$ is independent of $n$ for $n\gg1$, the entanglement depends only on $\theta_k=\frac{\pi}{2}k/k_{opt}$ and not on $n$. Namely, it is scale invariant, which has been discussed in \cite{PRA87}. The maximum entanglement asymptotically converges to $0.5$, which is obtained when $\theta_{k_T}=\pi/4$, $i.e.$, the iteration number is just $k_T/k_{opt}=0.5$.
\end{Remark}

\subsection{GHZ state}
A GHZ state \cite{07120921} is defined as
\begin{equation}
\ket{GHZ}=\frac{1}{\sqrt2}(\ket{00...0}+\ket{11..1}),
\end{equation}
which is a generalization of a Bell state.
If the state $\ket{S_1}$ is a GHZ state, then the number of marked states is $M=2$, and the marked states are $\ket{00...0}$ and $\ket{11...1}$. Thus, we have $B(\alpha)=\frac{1}{\sqrt 2} (\cos^{n}\frac{\alpha}{2}+\sin^n\frac{\alpha}{2})$.
After $k$ iterations, the state can be expressed as $\ket{\psi_{k,M=2}}$, and the overlap between this state and the separate state $\ket{\phi}$ is
\begin{eqnarray}\label{GHZlap}
\braket{\psi_{k,M=2}|\phi}=&&
\frac{\cos\theta_k}{\sqrt{N-2}}\left[\left(\cos{\frac{\alpha}{2}}+\sin{\frac{\alpha}{2}}\right)^n
-\left(\cos^n\frac{\alpha}{2}+\sin^n\frac{\alpha}{2}\right)\right]\nonumber\\
&&+\frac{\sin\theta_k}{\sqrt{2}}\left(\cos^{n}\frac{\alpha}{2}+\sin^n\frac{\alpha}{2}\right).
\end{eqnarray}

The optimization of $B=\frac{1}{\sqrt 2} (\cos^{n}\frac{\alpha}{2}+\sin^n\frac{\alpha}{2})$ is obtained when $\alpha=0$ (for $n\geq2$) or $\alpha=\pi/2$ (for $n=1$). Therefore, $B_{max}=1$, which is independent of $n$ obtained when $\alpha=0$ for $n\geq2$.
According to Eq. (\ref{thk}), the turning point is $\theta_{k_T}=\arctan\sqrt2$, and $k_T\simeq 0.61k_{opt}$.
When $n\gg1$, the GME of $\ket{\phi_{k,M=2}}$ can be expressed as follows:
\begin{eqnarray}
E(\ket{\psi_{k,M=2}})\simeq\label{GHZr}
\left\{
\begin{aligned}
&\sin^2\theta_k,& \theta_k\leq \arctan\sqrt2;\\
&\frac{1+\cos^2\theta_k}{2},&\theta_k > \arctan\sqrt2.
\end{aligned}
\right.
\end{eqnarray}

\begin{figure}
\includegraphics[width=0.8\textwidth]{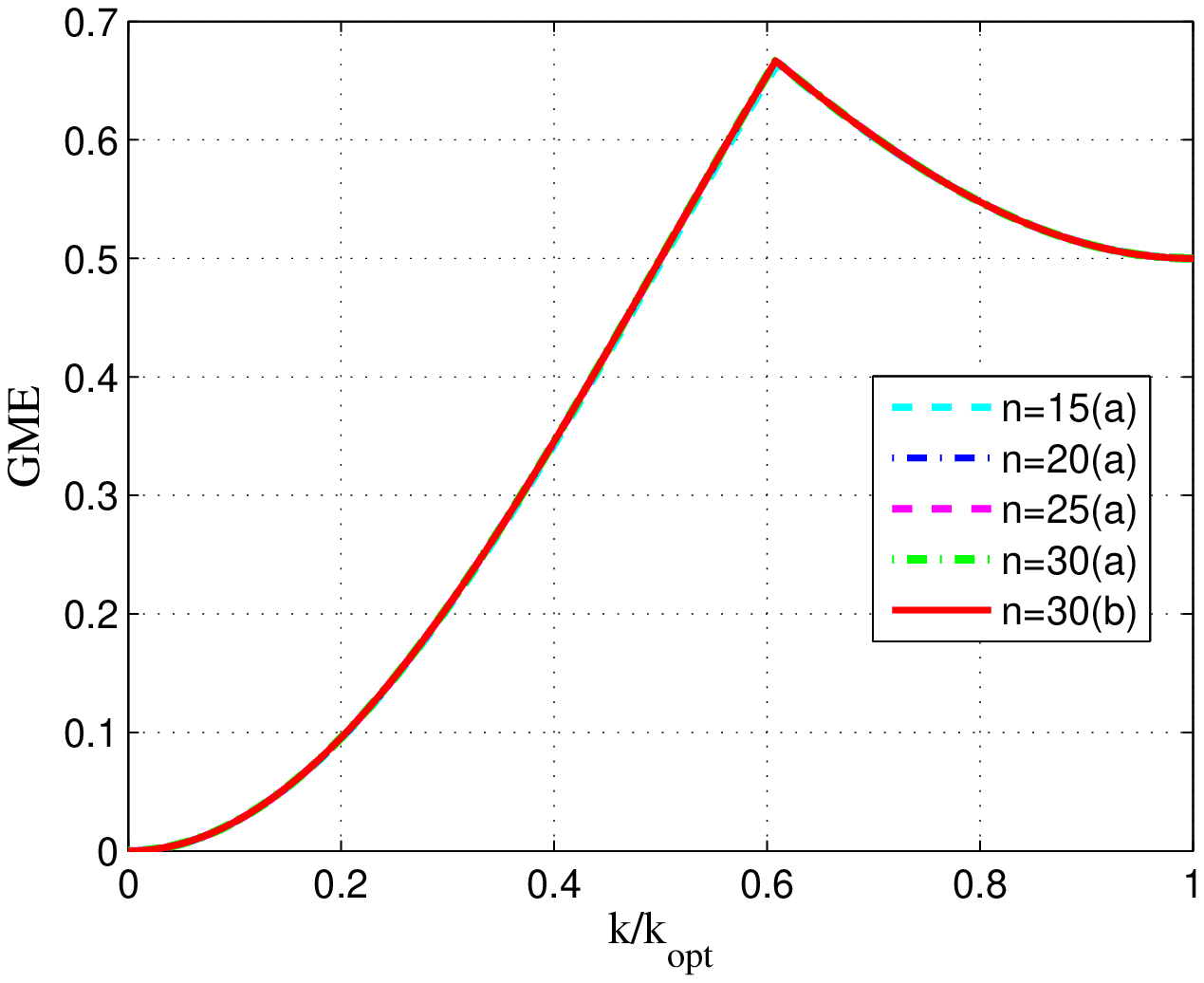}
\caption{\label{fig:GHZ} (Color online) Entanglement dynamics vs $k/k_{opt}$ when the marked states are $GHZ$ states for $n=15\sim30$. Here, (a) means that the curves of the GME are plotted by Eq. (\ref{defEn}) and Eq. (\ref{GHZlap}), and (b) is plotted by Eq. (\ref{GHZr}).}
\end{figure}

This expression shows that the GME is scale invariant and only depends on $\theta_k$, $i.e.$, just depends on $k/k_{opt}$ when $\ket{S_1}$ is a GHZ state.
To qualify the dynamics of entanglement change, we present the numerical results in Fig. \ref{fig:GHZ}. We can observe that the curves only change with $k/k_{opt}$ and do not depend on $n$, showing scale invariance. Furthermore, the curves of the GME plotted by Eq. (\ref{defEn}), Eq. (\ref{GHZlap}) and Eq. (\ref{GHZr}) are identical for $n=30$.

\newtheorem{Remark2}{\textbf{Remark}}\label{R2}
\begin{Remark}
In the case where $n\gg1$, according to Eq. (\ref{GHZr}), the GME also exhibits scale invariance, as discussed by Rossi {\em et al.} \cite{PRA87}. In this case, the optimization of $B(\alpha)=\frac{1}{\sqrt 2}(\cos^{n}\frac{\alpha}{2}+\sin^n\frac{\alpha}{2})$ is independent of $n$. The maximum entanglement asymptotically converges to $0.667$ when $\theta_{k_T}=\arctan\sqrt2$, i.e., the turning iteration is $k_T\simeq 0.61 k_{opt}$.
\end{Remark}

\subsection{$W$ state}
Dicke states \cite{Dicke}, which are also named the $W$ family, are an important family of symmetric states. A $W$ state \cite{PRA62} is a particular member of the family, which is defined as
\begin{equation}
\ket{W}=\frac{1}{\sqrt n}(\ket{00...01}+\ket{00...10}+...+\ket{10...00}).
\end{equation}
If the marked state $\ket{S_1}$ is a $W$ state, the number of marked states is $M=n$, and the Hamming weights of $M$ marked states are $1$ ({\em i.e.}, $n_i=1,\ i=1,2,...,M$). Thus, $B(\alpha)=\sqrt n \cos^{n-1}\frac{\alpha}{2}\sin\frac{\alpha}{2}$.
The state after $k$ iterations can be expressed as $\ket{\psi_{k,M=n}}$, and the overlap between this state and the separable state $\ket\phi$ can be expressed analytically as
\begin{eqnarray}\label{Wlap}
\braket{\psi_{k,M=n}|\phi}=&&
\frac{\cos\theta_k}{\sqrt{N-n}}\left[\left(\cos{\frac{\alpha}{2}}+\sin{\frac{\alpha}{2}}\right)^n
-n \cos^{n-1}\frac{\alpha}{2}\sin\frac{\alpha}{2}\right]\nonumber\\
&&+\sin\theta_k \left(\sqrt n \cos^{n-1}\frac{\alpha}{2}\sin\frac{\alpha}{2}\right).
\end{eqnarray}
The maximum of $|B(\alpha)|^2$ is
$$\max_\alpha|B(\alpha)|^2=n \max_\alpha\left|\cos^{n-1}\frac{\alpha}{2}\sin\frac{\alpha}{2}\right|^2
=\left(1-\frac{1}{n}\right)^{n-1},$$
which is obtained when $\alpha=2\arccos\sqrt{1-\frac{1}{n}}$.

 When $n\gg 1$, the entanglement can be expressed as
\begin{flalign}\label{Wr}
E(\ket{\psi_{k,M=n}})\simeq
\left\{
\begin{aligned}
&\sin^2\theta_k, &\theta_k\leq \theta_{k_T};\\
&1-\sin^2\theta_k\left(1-\frac{1}{n}\right)^{n-1}, &\theta_k > \theta_{k_T},
\end{aligned}
\right.
&
\end{flalign}
where the turning point is
\begin{eqnarray}
\theta_{k_T}=\arctan\left(\left(1-\frac{1}{n}\right)^\frac{1-n}{2}\right).
\end{eqnarray}

The GME is scale invariant when $\theta_k\leq\theta_{k_T}$. However, in addition to depending on $\theta_k$, the GME also depends on the number $n$ of qubits when $\theta_k > \theta_{k_T}$. Therefore, it does not have scale invariance in the entire search process.

We also present numerical results to qualify the dynamics of entanglement change for when $\ket{S_1}$ is a W state in Fig. \ref{fig:W}.
We can observe that the turning points  $k_T/k_{opt}=2\theta_{k_T}/\pi$ exist during the curves of the GME, which depend on $n$. The curves are first increasing and then decreasing. They are almost identical before the turning points, which only change with $k/k_{opt}$, showing scale invariance. However, the curves of the GME after the turning points not only change with $k/k_{opt}$ but also with $n$, which means that the GME does not have scale invariance. The result is consistent with our analytic result. Furthermore, the curves of the GME plotted by Eqs. (\ref{defEn}), (\ref{Wlap}) and (\ref{Wr}) are almost the same when $n=35$.

\newtheorem{Remark3}{\textbf{Remark}}\label{R3}
\begin{Remark}
 From the above discussions, the GME in Grover's search algorithm does not have scale invariance in the entire search process because it depends on $n$ after the turning point $k_T/k_{opt}$. Note that $\max_\alpha|B(\alpha)|^2\simeq 1/e$ in the limit $n\rightarrow\infty $. Hence, the position of the turning iteration is $k_T/k_{opt}\simeq0.653$, and the maximum entanglement is $E_{\max}=\sin^2\theta_{k_T}\simeq0.73$. The GME of the final state is $E_{opt}=E(W)=1-(\frac{n-1}{n})^{n-1}\simeq0.632$. In this respect, the GME is scale invariant in the limit $n\rightarrow\infty $.
\end{Remark}

\begin{figure}
\includegraphics[width=0.8\textwidth]{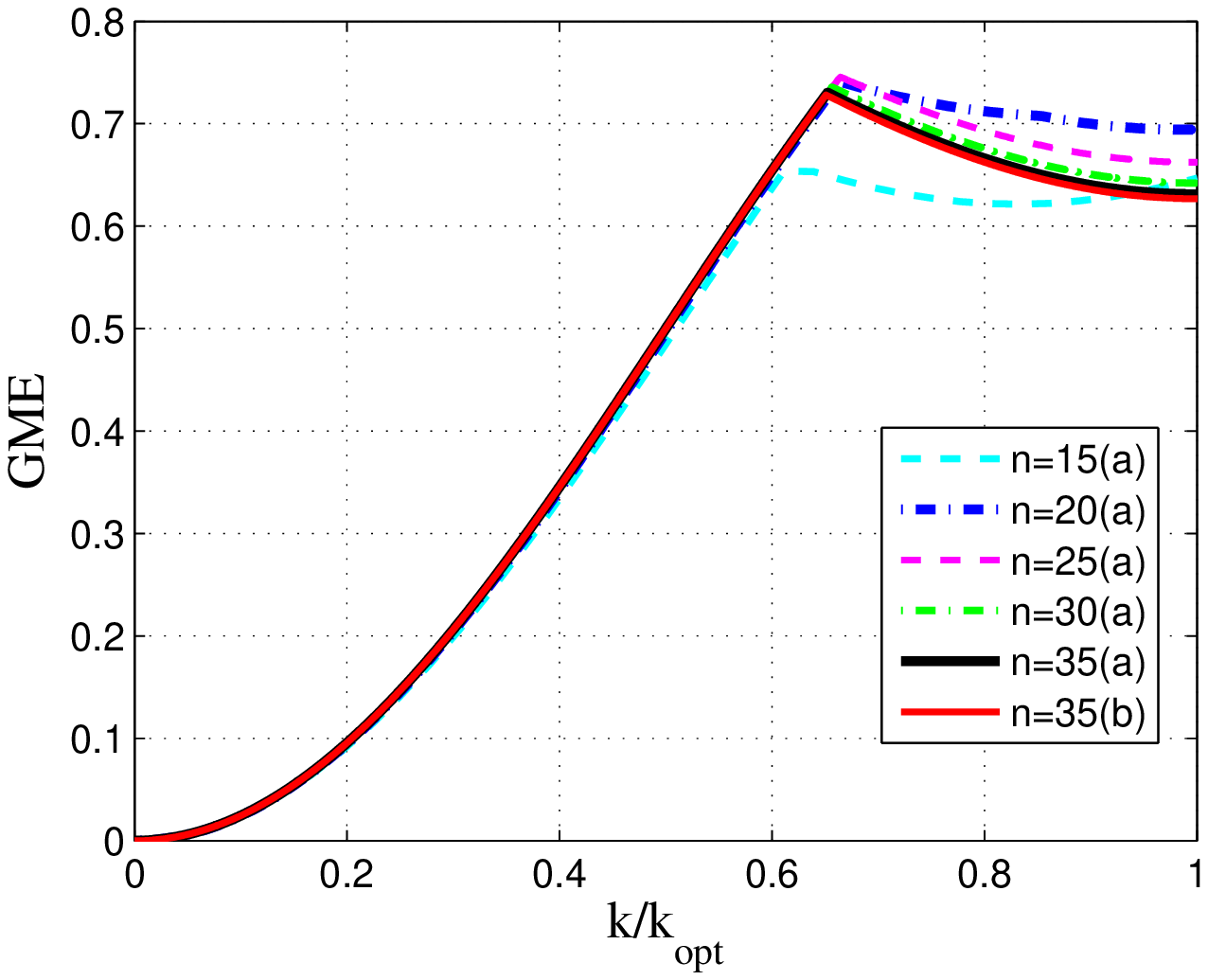}
\caption{\label{fig:W} (Color online) Entanglement dynamics vs $k/k_{opt}$ when the marked states are $W$ states for $n=15\sim35$. Here, (a) means that the curves of the GME are plotted by Eq. (\ref{defEn}) and Eq. (\ref{Wlap}), and (b) is plotted by Eq. (\ref{Wr}).
}
\end{figure}

\subsection{Discussion}
Until now, we have discussed the cases where $\ket{S_1}$ are product, GHZ and W states both analytically and numerically for different $n$. Comparing Eq. (\ref{M1r}), Eq. (\ref{GHZr}) and Eq. (\ref{Wr}), we find that $\theta_{k_T}$ are different. The GMEs are all asymptotic to $\sin^2\theta_k$ for $\theta_k\geq \theta_{k_T}$, but they are different for $\theta_k<\theta_{k_T}$, which shows that they depend on the marked states. We also present the GMEs of them when $n=28$ in Fig. \ref{fig:MGW}, where we can observe that the curves of the GMEs are the same before the turning point $k_T/k_{opt}$ but different after $k_T/k_{opt}$, which are decided by the marked states. The turning points are also different. Therefore, the result is consistent with the analytic result.

Two more simple examples are the GMEs of the states produced by Grover's search algorithm when $\ket{S_1}=1/\sqrt 2(\ket{0...00}+\ket{0...01})$ and $\ket{S_1}=1/\sqrt 2(\ket{0...00}+\ket{1...11})$ (GHZ state). Since $1/\sqrt 2(\ket{0...00}+\ket{0...01})$ is an asymmetric product state, the GME curve is the same as the case when $\ket{S_1}=\ket{0...00}$ according to section $4.1$. Thus, the turning point is $k_{opt}/2$, and the GME after turning is $cos^2 \theta_k$. Both of them are different from the case where $\ket{S_1}$ is a GHZ state. In these two cases, the difference comes from the Hamming weights of the marked states since the numbers of marked states are both $M=2$. Therefore, the Hamming weights of the marked states also play an important role in the calculation of the GME.

\begin{figure}
\includegraphics[width=0.8\textwidth]{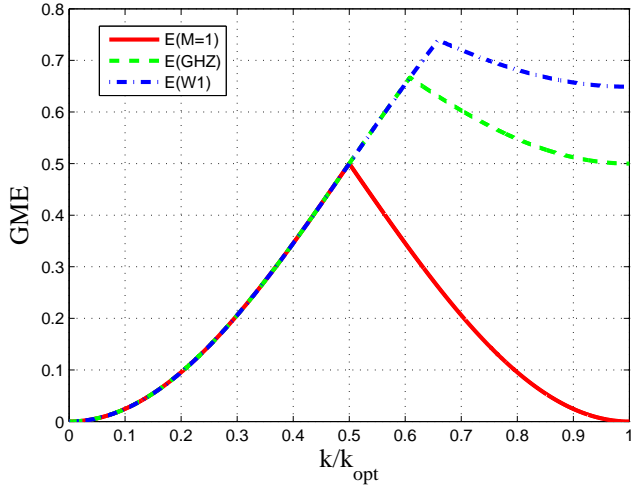}
\caption{\label{fig:MGW} (Color online) Entanglement dynamics vs $k/k_{opt}$ when $\ket{S_1}$ are product, GHZ and W states for $n=28$, which are plotted by Eq. (\ref{defEn}) with Eq. (\ref{M1lap}), Eq. (\ref{GHZlap}) and Eq. (\ref{Wlap}), respectively. 
}
\end{figure}

\section{Conclusion} \label{Sec5}
Using the geometric measure of entanglement (GME), the amount of global multipartite entanglement after each iteration can be quantified in Grover's algorithm.
In this paper, we have considered $M$ marked symmetric states in a database of size $N=2^n$ when $N\gg M$. We first discussed the optimization process to effectively compute the GME. Then, we presented the GME expression for when the entanglement behaves asymptotically for large $N$. We have shown that  a turning point $\theta_{k_T}$ always exists in the GME curve and deduced a general formula to compute the turning point in terms of the number of the marked states and their Hamming weights. Before the turning point, the entanglement is always asymptotic to $\sin^2\theta_k$, which only depends on the ratio of the iteration number $k$ to the total iteration $k_{opt}$, $i.e.$, $k/k_{opt}$. However, the entanglement after the turning point often also depends on $n$ and the marked states (the number of the marked states and their Hamming weights). We also provide a sufficient and necessary condition when the GME is scale invariant.
To clearly illustrate the above conclusions, we presented both analytical expressions and numerical results when the marked states were product, GHZ and W states.

To summarize, we have answered the question in the introduction. ``Scale invariance" of the global multipartite entanglement dynamics holds only for some special cases when $N\gg M$. In general, the entanglement dynamics is not scale invariant in terms of the GME because it typically depends on the  number $n$ of qubits and the marked states during the process of the search algorithm. However, the GME is asymptotic to $\sin^2\theta_k$ before the turning point, which partially shows scale invariance.
In this paper, we only discuss the superposition state of $M$ marked states being symmetric. Do the conclusions hold in the case where the state is asymmetric? Can we obtain similar conclusions using other measurements of entanglement? These questions may be worth studying further.

\begin{acknowledgements}
We are thankful to the anonymous referees and editor for their comments and suggestions that have greatly helped to improve the quality of the manuscript. This work is supported in part by the National Natural Science Foundation of China (Nos. 61572532, 61272058, 61602532) and the Fundamental Research Funds for the Central Universities of China (Nos. 17lgjc24, 161gpy43).
\end{acknowledgements}



\begin{thebibliography}{}
\bibitem{book1} Nielsen, M.A., Chuang, I.L.: Quantum computation and quantum information. Cambridge University Press. Cambridge (2000)
\bibitem{JMP43} Bru{\ss}, D.: Characterizing entanglement. J. Math. Phys. \textbf{43}(9), 4237-4251 (2002)
\bibitem{Shor} Shor, P.W.: Algorithms for quantum computer: discrete logarithm and factoring.
   in Proceedings of the 35th IEEE Symposium on the Foundations of Computer Science,
   IEEE Computer Society, Los Alamitos, CA, pp. 56-65 (1994)
\bibitem{Grover97} Grover, L.K.: A fast quantum mechanical algorithm for database search.Phys. Rev. Lett. \textbf {79}, 325 (1997). 
\bibitem{PRSLA459} Jozsa, R., Linden, N.: On the role of entanglement in quantum-computational speed-up. Proc. R. Soc. Lond. A \textbf{459}, 2011-2023 (2003)
\bibitem{PRL91} Vidal, G.: Efficient Classical Simulation of Slightly Entangled Quantum Computations, Phys. Rev. Lett. \textbf{91}, 147902 (2003)
\bibitem{PRSLA439} Deutsch, D., Jozsa, R.: Rapid solution of problems by quantum computation.
    Proc. R. Soc. Lond. A \textbf{439}, 553 (1992)
\bibitem{Simon94} Simon, D.: On the power of quantum computation. SIAM J. Comput. {\bfseries 26}, 1474--1483 (1997).  Earlier version in FOCS'94.
\bibitem{PRA83} Bru{\ss}, D., Macchiavello, C.: Multipartite entanglement in quantum algorithms. Phys. Rev. A \textbf{83}, 052313 (2011)
\bibitem{PRL85} Meyer, D.A.: Sophisticated Quantum Search Without Entanglement. Phys. Rev. Lett.
    \textbf{85}, 2014 (2000)
\bibitem{PRA65} Biham, O., Nielsen, M.A., Osborne, T.J.: Entanglement monotone derived from Grover¡¯s algorithm. Phys. Rev. A \textbf{65}, 062312 (2002)
\bibitem{PRA75} Shimoni, Y., Biham, O.: Groverian entanglement measure of pure quantum states with arbitrary partitions. Phys. Rev. A \textbf{75}, 022308 (2007)
\bibitem{IJQI6} Wen, J.Y., Qiu, D.W.: Entanglement in adiabatic quantum searching algorithms. Int. J. Quantum Inf. \textbf{6}(5), 997--1009 (2008)
\bibitem{JMP439} Meyer, D. A.,  Wallach, N. R. Global entanglement in multiparticle systems. Journal of Mathematical Physics, \textbf{43}(9), 4273-4278 (2002).
\bibitem{PRA69}Shimoni, Y., Shapira, D., Biham, O.: Characterization of pure quantum states of multiple qubits using the Groverian entanglement measure. Phys. Rev. A \textbf{69}, 062303 (2004)
\bibitem{PLA345} Fang, Y., Kaszlikowski, D., Chin, C., Tay, K., Kwek, L.C., Oh, C.H.: Entanglement in the Grover search algorithm. Phy. Lett. A \textbf{345}, 265-272 (2005).
\bibitem{PLA373} Rungta, P.: The quadratic speedup in Grover's search algorithm from the entanglement perspective. Phys. Let. A \textbf{373}, 2652-2659 (2009)
\bibitem{JPMT43} Cui, J., Fan, H.: Correlations in Grover search. Journal of Physics A Mathematical and Theoretical \textbf{43}(4), 045305 (2010).
\bibitem{PRA87} Rossi, M., Bru{\ss}, D., Macchiavello,C.: Scale invariance of entanglement dynamics     in Grover's quantum search algorithm, Phys. Rev. A \textbf{87}, 022331 (2013)
\bibitem{13054454} Chakraborty, S., Banerjee, S., Adhikari, S., Kumar, A.: Entanglement in the Grover's Search Algorithm. arXiv: 1305.4454 (2013)
\bibitem{Nat14} Qu, R., Shang, B.j., Bao, Y.R., Song, D.W., Teng, C.M., Zhou, Z.W.: Multipartite entanglement in Grover's search algorithm. Nat. Comput. \textbf{14}: 683-689 (2015).
\bibitem{QIP152} Batle, J., Raymond Ooi, C. H., Farouk, A., Alkhambash, M.S., Abdalla, S.: Global versus local quantum correlations in the Grover. Quantum Information Processing \textbf{15}(2), 833-850 (2016)
\bibitem{QIP1511}Holweck, F., Jaffali, H., Nounouh, I. Grover's algorithm and the secant varieties. Quantum Information Processing, \textbf{15}(11), 4391-4413 (2016).
\bibitem{ANYA755} Shimony, A.: Degree of entanglement.  Annals of the New York Academy of Sciences \textbf{755}, 675-679 (1995)
\bibitem{PRA68} Wei, T.C., Goldbart, P.M.: Geometric measure of entanglement and applications to bipartite and multipartite quantum states. Phys. Rev. A \textbf{68}, 042307 (2003)
\bibitem{PRA77} Blasone, M., Dell'Anno, F., Siena, S.D., Illuminati, F.: Hierarchies of geometric entanglement. Phys. Rev. A \textbf{77}, 062304 (2008).
\bibitem{PRA73} Cavalcanti, D.:Connecting the generalized robustness and the geometric measure of entanglement. Phys. Rev. A \textbf{73}, 044302 (2006)
\bibitem{PRL98} G¨¹hne, O., Reimpell, M., Werner, R.F.:Estimating Entanglement Measures in Experiments. Phys. Rev. Lett. \textbf{98}, 110502 (2007).
\bibitem{07120921} Greenberger, D.M., Horne, M.A., Zeilinger, A.: Going beyond Bell's theorem, arXiv: 0712.0921 (2007)
\bibitem{PRA81}Martin J. , Giraud O., Braun, P. A., Braun, D., Bastin T.: Multiqubit symmetric states with high geometric entanglement.Phys. Rev. A \textbf{81}, 062347 (2010)
\bibitem{PRA80}H$\ddot{u}$bener, R., Kleinmann, M., Wei, T.C., Gonz$\acute{a}$lez-Guill$\acute{e}$n, C., G$\ddot{u}$hne, O.: Geometric measure of entanglement for symmetric states.
    Phys. Rev. A \textbf{80}, 032324 (2009)
\bibitem{PRA67} Stockton, J.K., Geremia, J.M., Doherty, A.C., Mabuchi, H.: Characterizing the entanglement of symmetric many-particle spin-$\frac{1}{2}$ systems. Phys. Rev. A \textbf{67}, 022112 (2003)
\bibitem{Dicke} Dicke, R.H.: Coherence in Spontaneous Radiation Processes. Phys. Rev. \textbf{93}, 99 (1954)
\bibitem{PRA62} Dur, W., Vidal, G., Cirac, J.I.: Three qubits can be entangled in two inequivalent ways. Phys. Rev. A \textbf{62}, 062314 (2000)
\bibitem{PLA346} Chamoli, A., Bhandari, C.M.: Success rate and entanglement measure in Grover's search algorithm for certain kinds of four qubit states. Phys. Lett. A \textbf{346}, 17-26 (2005)

\end{thebibliography}
\end{document}